\documentclass{article}
\usepackage{llncsdoc}

\usepackage{fullpage}
\usepackage{color}
\definecolor{darkgreen}{rgb}{0.1,0.4,0.1}
\usepackage[pagebackref,letterpaper=true,colorlinks=true,pdfpagemode=none,linkcolor=blue,citecolor=darkgreen,pdfstartview=FitH]{hyperref}

\usepackage{graphicx}
\usepackage{amssymb}
\usepackage{amsmath,algorithm,algorithmic,amsthm}

\newtheorem{theorem}{Theorem}
\newtheorem{lemma}[theorem]{Lemma}

\newtheorem{corollary}[theorem]{Corollary}

\theoremstyle{definition}

\theoremstyle{remark}
\newtheorem{remark}[theorem]{Remark}

\def\OPT{\textup{OPT}}

\def\D{\mathcal{D}}
\newcommand{\PP}[2]{{\bf P}_{#1}\left[#2\right]}
\newcommand{\EE}[2]{{\bf E}_{#1}\left[#2\right]}

\def\eps{\epsilon}
\def\lpacss{\textup{(LP-ACSS)}}
\newenvironment{proofof}[1]{{\noindent \em Proof of #1.  }}{\hfill\qed}
\renewenvironment{proof}{{\noindent \em Proof.  }}{\hfill\qed}

\begin{document}

\title{A Rounding by Sampling Approach to the\\ Minimum Size $k$-Arc Connected Subgraph Problem}

\author{
Bundit Laekhanukit\thanks{
School of Computer Science, McGill University. Email: \protect\url{blaekh@cs.mcgill.ca}
}
\and
Shayan Oveis Gharan\thanks{
Department of Management Science and Engineering, Stanford University.
Supported by a Stanford Graduate Fellowship. Email: \protect\url{shayan@stanford.edu}.
}
\and
Mohit Singh\thanks{
McGill University and Microsoft Research, Redmond. Email: \protect\url{mohit@cs.mcgill.ca}
}
}


\maketitle

\begin{abstract}
In the $k$-arc connected subgraph problem, we are given a directed graph $G$ and an integer $k$ and the goal is the find a subgraph of minimum cost such that there are at least $k$-arc disjoint paths between any pair of vertices. We give a simple $(1+1/k)$-approximation to the unweighted variant of the problem, where all arcs of $G$ have the same cost. This improves on the $1+2/k$ approximation of Gabow et al.~\cite{GGTW09}.

Similar to the 2-approximation algorithm for this problem~\cite{FJ81}, our algorithm simply takes the union of a $k$ in-arborescence and a $k$ out-arborescence. 
The main difference is in the selection of the two arborescences. Here, inspired by the recent applications of the rounding by sampling method (see e.g. \cite{AsadpourGMGS10,MOS11,GharanSS11,AnKS11}),
we select the arborescences randomly by sampling from a distribution on unions of $k$ arborescences that is defined based on an {\em extreme point solution} of the linear programming relaxation of the problem. In the analysis, we crucially utilize the sparsity property of the extreme point solution to upper-bound the size of the union of the sampled arborescences.

To complement the algorithm, we also show that the integrality gap of the minimum cost strongly connected subgraph problem (i.e., when $k=1$) is at least $3/2-\eps$, for any $\eps>0$. Our integrality gap instance  is inspired by the integrality gap example of the asymmetric traveling salesman problem~\cite{CGK06}, hence providing further evidence of connections between the approximability of the two problems.
\end{abstract}

\section{Introduction}

In the {\em minimum cost $k$-arc connected spanning subgraph}
(min-cost $k$-ACSS) problem, we are given a directed graph
$G=(V,A)$ with cost $c:A\rightarrow R$ on the arcs and a connectivity requirement $k$.
The goal is to find a spanning subgraph $G'=(V,A')$ of $G$ of
minimum total cost which is {\em $k$-arc connected}, i.e., every pair of vertices have at least $k$-arc disjoint paths between them. The special case of $k=1$, $1$-ACSS problem, is called the {\em minimum cost strongly connected subgraph} problem.
In the unweighted variant of $k$-ACSS, the {\em minimum size $k$-arc connected spanning subgraph} (min-size $k$-ACSS) problem, where all arcs of $G$ have the same cost, we want to minimize the number of arcs that we choose.

 The min-cost $k$-ACSS problem has a $2$-approximation algorithm~\cite{FJ81}, and it has been a long standing open problem  to improve this bound. Significant attention has been given to the unweighted variant of the problem. In particular, the minimum size strongly connected subgraph problem is very well studied~\cite{FJ81,KRY94,KRY96,Vetta01,ZNI03}, and  the current best approximation ratio is $3/2$, which is due to Vetta~\cite{Vetta01}. The min-size $k$-ACSS problem has been shown to be easier as $k$ increases~\cite{CT00,Gabow04,GGTW09}, and the  best approximation ratio is $1+2/k$ that is given in the work of Gabow et al.~\cite{GGTW09}.
This approximation ratio is almost tight as the min-size $k$-ACSS problem does not admit $(1+\epsilon/k)$-approximation, for some fixed $\epsilon>0$, unless P=NP~\cite{GGTW09}.
Similar to the directed case, the {\em minimum size $k$-edge connected subgraph spanning} problem, an undirected variant of the min-size $k$-ACSS problem, is known to be easier as $k$ increases, and the best known approximation ratio for this problem is $1+1/(2k)+O(1/k^2)$ due to Gabow and Gallagher~\cite{GabowG08}.

\subsection {Our Results}
In this paper, we give improved upper and lower bounds for the $k$-ACSS problem. We first show the following improved algorithms for the min-size $k$-ACSS problem.

\begin{theorem}\label{thm:main-upper}
For any $k\geq 1$, there is a $\min\{7/4,1+1/k\}$-approximation algorithm for the min-size $k$-ACSS problem. 
\end{theorem}

Similar to the simple 2-approximation algorithm for the minimum-cost $k$-ACSS problem, our algorithm takes the union  of  a $k$ in-arborescence and a $k$ out-arborescence.
The main difference is in the selection of the two arborescences. Here,
we select the arborescences randomly by sampling from a distribution on unions of $k$ arborescences that is defined by the linear programming relaxation of the problem. In particular,
we write a convex combination of the unions of $k$-arborescences such that the marginal probability of each arc is bounded above by its fraction in the solution of LP relaxation.

The algorithm essentially employs the rounding by sampling method that recently has been applied to various problems in the algorithm design and online optimization literature (c.f.~\cite{AsadpourGMGS10,MOS11,GharanSS11,AnKS11}), while the analysis is much simpler in our setting. Here, the main technical difference is a crucial use of the extreme point solutions of LP relaxation. In particular, because of the sparsity of the extreme point solutions, we can argue that
the union of $k$ in-arborescences and $k$ out-arborescences is not much larger than the size of the support of the LP extreme point solution and thus the size of the optimum.

Our result improves on the $(1+\frac2k)$-approximation of Gabow et al.~\cite{GGTW09} for the min-size $k$-ACSS problem, for any $k>0$. Furthermore, for the minimum size strongly connected subgraph problem, while we do not improve the approximation factor of $\frac32$~\cite{Vetta01}, our algorithm is much simpler and gives a possible direction for
weighted version of the problem.

To complement the positive results, we prove that the integrality gap of the natural linear programming relaxation of the strongly connected subgraph problem is bounded below by $3/2-\eps$ for any $\epsilon >0$.
\begin{theorem}\label{thm:main-lower}
For any  $\epsilon >0$, the integrality gap of the standard  linear programming relaxation for the minimum cost strongly connected subgraph problem is at least $\frac{3}{2}-\epsilon$.
\end{theorem}

To the best of our knowledge, there is no explicit construction that gives a lower bound on the integrality gap of the minimum cost strongly connected subgraph problem.
Our integrality gap example builds on a similar construction for the asymmetric traveling salesman problem~\cite{CGK06} and shows stronger connections between the two problems.

\subsection{Notations}
Let $\delta_{G}^+(U):=\{(u,v)\in E:u\in U, v\in V\setminus U\}$ denote
the set of arcs leaving $U$ in a graph $G$; if  $G$ is clear
in the context,  we will skip the subscript.

A graph $G$ is {\em $k$-arc connected}
if and only if every (proper) subset of vertices $U\subset V$ have at least $k$
leaving arcs, i.e., $|\delta^+_{G}(U)|\geq k$, and 
$G$ is {\em strongly connected} if it is $1$-arc connected.
We may drop the subscript if $G$ is clear in the context.
We use the following Linear Programming relaxation for $k$-ACSS.

\[
\begin{array}{llllll}
\mbox{(LP-ACSS)~~~~} & \mbox{minimize} & ~~~~& \displaystyle\sum_{a\in A}c_ax_a \\
\\
& \mbox{subject to}
 & & x(\delta^+(U)) 
 \geq k & ~~~~&  \forall U\neq \emptyset, U\subsetneq V \\
 \\
& & & 0\leq x_a\leq 1 & & \forall e\in E,
\end{array}
\]
\label{LP:kecss}
where $x(\delta^+(U)) = \sum_{a\in \delta^+(U)} x_a$. Throughout the paper $x$ will always be an optimum solution of the \lpacss .

For any vector $y:A\rightarrow \mathbb{R}$, and a set $F\subset A$ of arcs, $y(F):=\sum_{a\in F} y_a,$ is the sum of the values of the arcs in $F$, and $c(F):=\sum_{a\in F} c_a$ is the
sum of the cost of the arcs in $F$.
Also, $\chi(F)$ denotes the characteristic vector
of the set $F$, i.e., $\chi(F)_a=1$ if $a\in F$ and $\chi(F)_a=0$ otherwise.

\section{An Approximation Algorithm for Min-Size $k$-ACSS}
In this section, we prove Theorem \ref{thm:main-upper}: given a graph $G$, we give a polynomial time algorithm that finds a $k$-arc connected subgraph of $G$ such that it has no more than $\min\{1+1/k, 7/4\}$ of the arcs of the optimum solution. 
Before describing the algorithm, we need to recall some of the properties of arborescences in directed graphs.

Given a directed graph $G$ and a (root) vertex $r\in V$, an {\em $r$-out arborescence} $T$ of $G$ is a directed tree rooted at $r$ that contains a path from $r$ to every other vertex of $G$. 
An {\em $r$-out $k$-arborescence} is a subgraph $T$ of $G$ that is the union of $k$ arc-disjoint $r$-out arborescences. An {\em $r$-in arborescence} and an {\em $r$-in $k$-arborescence} are defined analogously. 
The following polyhedron plays an important role in the design and analysis of our algorithm.

\[
P^{out}=\left\{y : y(\delta^+(U)) \geq k, \quad
               \forall \emptyset\neq U\subsetneq V\setminus\{r\}, 0\leq y \leq 1
        \right\}
\]

Frank~\cite{Frank79} showed that $P^{out}$ is the up hull of the convex hull  of $r$-out $k$-arborescences (see Corollary 53.6a~\cite{schrijver}), and it can be seen that every feasible solution of (LP-ACSS) is a point in $P^{out}$.
Vempala and Carr~\cite{CV02} gave a polynomial-time algorithm that allows us to write a point $x\in P^{out}$ as a convex combination of $k$ arc-disjoint arborescences. Their algorithm requires a polynomial-time algorithm for finding an $r$-out $k$-arborescences~\cite{Edmonds72,Gabow91a}.

\begin{lemma}\label{lem:frank79b}~\cite{Frank79,CV02,Edmonds72,Gabow91a}
$P^{out}$ is the convex hull of subsets of $A$ containing  $r$-out $k$-arborescences. Moreover, given any fractional solution $y\in P^{out}$, there is a polynomial time algorithm that finds a convex combination of $r$-out $k$-arborescences,  $T_1,\ldots,T_l$, such that 
$$y\geq \sum_{i=1}^l \lambda_i \chi({T_i}).$$
\end{lemma}

The above lemma holds analogously for the $r$-in arborescences. Now, since $x\in P^{out}$, we can write a distribution of $r$-out(in) $k$-arborescences 
such that probability of each arc $a\in A$ chosen in a random $k$-arborescence is bounded above by $x_a$:
\begin{corollary}
\label{cor:arbsampling}
There are  distributions $\D_{in}(r)$ and $\D_{out}(r)$ of $r$-in $k$-arborescences and $r$-out $k$-arborescences, such that the marginal value of each arc $a\in A$ is bounded above by $x_a$, i.e., for all arcs $a\in A$,
\begin{eqnarray*}
\PP{T\sim \D_{in}(r)}{a\in T} &\leq & x_a,\\
\PP{T\sim \D_{out}(r)}{a\in T} &\leq & x_a.
\end{eqnarray*}
Moreover, these distributions can be computed in polynomial time. 
\end{corollary}
%
%

%
%
%
%

Now, we are ready to describe our algorithm. We sample $k$-arborescences $T_{in}$ and $T_{out}$ independently from $\D_{in}$ and $\D_{out}$, respectively, and we then return $T_{in}\cup T_{out}$ as an output. The details are described in Algorithm \ref{alg:main-algorithm}.

\begin{algorithm}[htb]
\caption{Approximation Algorithm for Min-Size $k$-ACSS}
\label{alg:main-algorithm}
\begin{algorithmic}[1]

\STATE Solve \lpacss~to get an optimum \emph{extreme point} solution $x$.


\STATE Find distributions $\D_{in}(r)$ and $\D_{out}(r)$ on $r$-in and $r$-out $k$-arborescences, respectively, such that the marginal value of each arc $a\in A$ is bounded above by $x_a$.
\STATE Sample an $r$-in $k$-arborescence $T_{in}$ from $\D_{in}(r)$ and an $r$-out $k$-arborescence $T_{out}$, {\em independently}, from $\D_{out}(r)$.

\RETURN  $T_{in}\cup T_{out}$.
\end{algorithmic}
\end{algorithm}

%

Next, we show that the approximation ratio of the above algorithm is no more than $1+1/k$.

\begin{theorem}
For any directed graph $G$, Algorithm \ref{alg:main-algorithm} always produces a $k$-arc connected subgraph of $G$ such that the expected size of the solution is no more than $\min\{7/4,1+1/k\}$ of the optimum.
\end{theorem}
\begin{proof}
First, we show that the union of any pair of $r$-in and $r$-out $k$-arborescences is $k$-arc connected. 
Let $T_{in}(T_{out})$ be a $r$-in ($r$-out) $k$-arborescence, and $H=T_{in}\cup T_{out}$. 
Since both $T_{in}$ and $T_{out}$ are unions of $k$ arc-disjoint arborescences, there are $k$ arc-disjoint paths from each of the vertices to $r$ and $k$ arc-disjoint paths from $r$ to each of the vertices.
Therefore,  $H$ remains strongly connected after removing any set of $k-1$ arcs. 
Hence, $H$ is $k$-arc connected. 

It remains to show that the expected size of the solution  is no more than $\min\{1+1/k,7/4\}$ of the optimum, i.e.,
$$ \frac{\EE{T_{in} \sim \D_{in}(r), T_{out}\sim \D_{out}(r)}{\left|T_{in} \cup T_{out}\right|}}{|\OPT|} \leq \min\left\{\frac74,1+\frac{1}{k}\right\}.$$

To simplify the notation, we will skip the subscript and write $\EE{}{\left|T_{in} \cup T_{out}\right|}$ to mean $\EE{T_{in} \sim \D_{in}(r), T_{out}\sim \D_{out}(r)}{\left|T_{in} \cup T_{out}\right|}$. Similarly, we will skip the subscripts for $\PP{T_{in}\sim \D_{in}(r)}{a\in T_{in}}$ and $\PP{T_{out}\sim \D_{out}(r)}{a\in T_{out}}$.

Since $T_{in}$ and $T_{out}$ are chosen independently,
\begin{eqnarray*}
\EE{}{\left|T_{in} \cup T_{out}\right|} &=& \sum_{a\in A} \left\{\PP{}{a\in T_{in}}+\PP{}{a\in T_{out}} 
- \PP{}{a\in T_{in}}\cdot\PP{}{a\in T_{out}}\right\}\\
&\leq & \sum_{a\in A} 2x_a - \sum_{a\in A} x_a^2.
\end{eqnarray*}

The last inequality follows from Corollary \ref{cor:arbsampling} and the fact that $x_a\leq 1$ for all $a\in A$. Let $F:=\{a :  0 < x_a < 1\}$ be the set of the fractional arcs (i.e., set of arcs with non-integer values in the solution of \lpacss). Since $x$ is an optimal solution of \lpacss, $|\OPT|\geq \sum_{a\in A} x_a$. Therefore,
\begin{eqnarray}
\frac{\EE{}{\left|T_{in} \cup T_{out}\right|}}{|\OPT|} &\leq& 1 + \frac{\sum_{a\in A} x_a - \sum_{a\in A} x_a^2}{\sum_{a\in A} x_a}\nonumber \\
&= & 1+ \frac{x(F) - \sum_{a\in F} x_a^2}{x(A)}\nonumber \\
&\leq & 1+ \frac{x(F) - x(F)^2/|F|}{x(A)}, \label{eq:ratioupperbound}
\end{eqnarray}
where the last inequality follows from Jenson's inequality and the fact that $f(t)=-t^2$ is a concave function. 

Since $x$ is an extreme point solution of \lpacss,~$x$ is a sparse vector.  It follows from the work of Melkonian and Tardos~\cite{MK04} (see also \cite{GGTW09}), that the number of fractional
arcs, $|F|$, is no more than $4n$. Hence,  
\begin{equation}
\label{eq:kbound} \frac{x(F) - x(F)^2/|F|}{x(A)} \leq \frac{x(F) - x(F)^2/4n}{x(A)} \leq \frac{n}{x(A)} \leq \frac{1}{k},
\end{equation}
where the second inequality follows since $x(F)-x(F)^2/4n$ attains its maximum at $x(F)=2n$, and the last inequality follows from the fact that $x(A) = \sum_{v\in V} x(\delta^+(v)) \geq nk$.
On the other hand, since $x(F)\leq x(A)$, we get
\begin{eqnarray}
\label{eq:1bound}
\frac{x(F) - x(F)^2/|F|}{x(A)} \leq \frac12 + \frac{x(F) - x(F)^2/2n}{2x(A)} \leq \frac12 + \frac{n}{4x(A)} \leq \frac34.
\end{eqnarray}
The theorem simply follows by putting  equations \eqref{eq:ratioupperbound},\eqref{eq:kbound},\eqref{eq:1bound} together.
\end{proof}

\medskip

\begin{remark}
Since the distributions $\D_{in}(r)$ and $\D_{out}(r)$ can be constructed such that the support of each distribution has size only polynomially large in $n$,  the algorithm can be derandomized simply by choosing a pair of $k$-arborescences that have the minimum number of arcs in their union.
\end{remark}

\section{A Lower Bound on the Integrality Gap}

In this section, we prove Theorem \ref{thm:main-lower}: we show a lower-bound of $1.5-\epsilon$, for any
arbitrary small $\epsilon>0$, on the integrality gap  of  \lpacss~for $k=1$. 
Our construction is based on the LP-gap construction of
the asymmetric traveling saleman problem by Charikar,
Goemans and Karloff~\cite{CGK06}.




\subsection{Construction}
Let $r>0$ be an integral parameter that will be defined later.
We start by defining the integrality gap example, $G(d,s,t)$, by a recursive construction
of depth $d$.  In any graph $G(d,s,t)$, $d$ is the {\em depth},
$r$ is the number of {\em columns}, $s,t$ are the {\em source, sink} vertices, respectively.
We allow $s$ and $t$ to be the same vertex. We will construct $G(d,s,t)$ inductively such that it contains exactly $r$ copies of
$G(d-1,.,.)$.

We start by describing $G(1,s,t)$. The graph consists of $s,t$ and $r$ distinct vertices
$v_1,\ldots,v_{r}$. Let  $v_0=s$ and $v_{r+1}=t$; note that $v_0$ and $v_{r+1}$ may be the same depending on the
given parameters $s$ and $t$.
For any $1 \leq i\leq r+1$, we include arcs $(v_i,v_{i-1})$ and
$(v_{i-1},v_i)$ in $G(1,s,t)$. Therefore,
$$ A(G(1,s,t)) := \{ (v_{i-1}, v_i), (v_{i}, v_{i-1}), 1\leq i\leq r+1\}.$$

Next, we define $G(d,s,t)$. 
The graph consists of $s,t$ and $r$ distinct copies of $G(d-1,.,.)$. In particular, let
$v_1,\ldots,v_{r},u_1,\ldots,u_{r}$ be $2r$ distinct vertices, and
$v_0=u_{r+1}=s$ and $v_{r+1}=u_0=t$.
For any $1\leq i\leq r$, include a distinct copy of $G(d-1,.,.)$ with source $u_i$ and sink $v_i$. Also,
 for any $1\leq i\leq r+1$, include the arcs $(v_i,v_{i-1})$ and
$(u_{i-1},u_i)$. Therefore,
$$ A(G(d,s,t)) := \{ (u_{i-1}, u_i), (v_{i}, v_{i-1}), 1\leq i\leq r+1\} \cup \left\{\bigcup_{i=1}^r A(G(d-1,u_i,v_i))\right\}. $$
 Figure \ref{fig:construction} illustrates the graph $G(3,s,s)$ for $r=3$.

Our  integrality gap example is $G_d:=G(d,s,s)$, where the source and the sink are unified.
The {\em $i^{th}$ column} of $G_d$ is defined to be the  $i^{th}$ copy of the $G(d-1,.,.)$,  i.e., $G_{d}^{(i)}:=G(d-1,u_i,v_i)$.
The set of arcs that connect the $r$ columns with $s$ and $t$, i.e.,  $A(G_d)\setminus\bigcup_{i=1}^rA(G_{d}^{(i)})$, are denoted
by {\em  $d^{th}$ level arcs}.
Similarly, the $l^{th}$ level arcs of $G_d$ are defined to be set of arcs included at the
$l^{th}$ level of induction. 
For example, the $(d-1)^{th}$ level arcs of $G_d$ are 
$\bigcup_{i=1}^{r}\left(A(G_{d}^{(i)})\setminus
\bigcup_{j=1}^rA(G_{d}^{(i;j)})\right)$, where $G_{d}^{(i;j)}$ is the
$j^{th}$ column of $G_{d}^{(i)}$.

We define the costs of the arcs of $G_d$ such that, for any $1\leq l\leq d$, the total cost of the arcs at level $l$ is equal to $1$. In other words, 
the cost of each arc at level $l$, $c_d(l)$, is  the reciprocal of the number of arcs at level $l$.
 By the construction of $G_d$, we have
\begin{eqnarray}
\label{eq:costfunction}
c_d(l) := \frac{1}{2(r+1) r^{d-l}}.
\end{eqnarray}

\begin{figure}
\begin{center}
\includegraphics[width=4.8in]{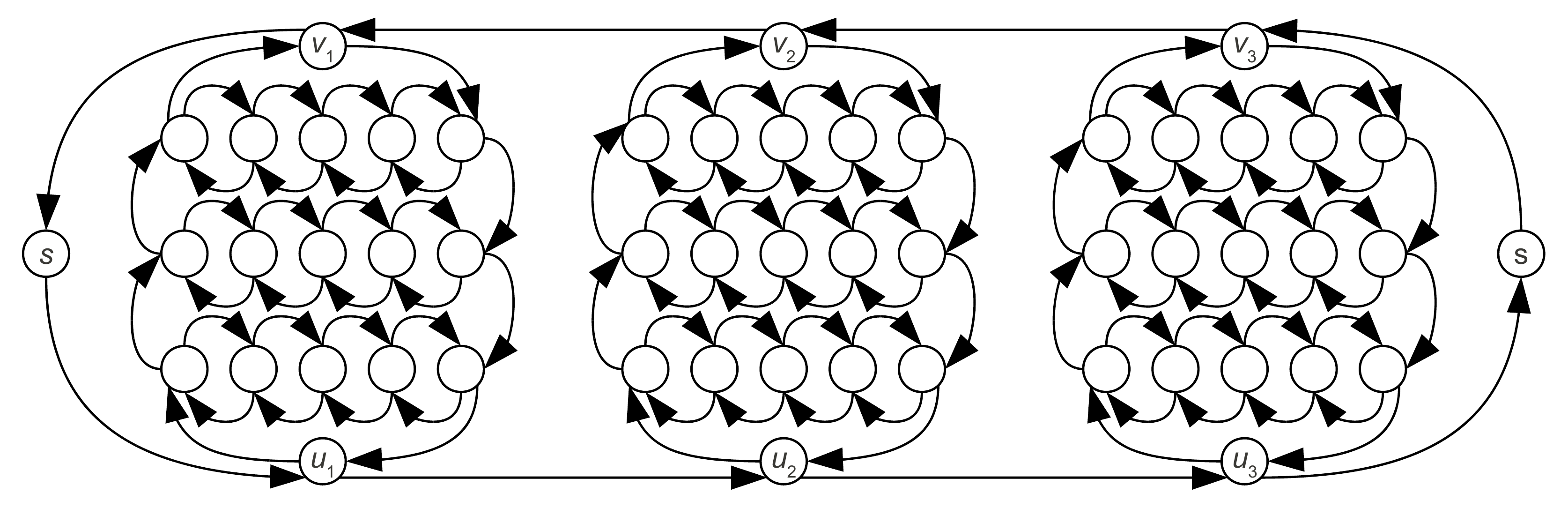}
\end{center}
\label{fig:construction}
\caption{An illustration of the graph $G(3,s,s)$, for $r=3$. Note that the vertices
  labeled ``s'' on the left and on the right are the same.
 }
\end{figure}

\subsection{Lower Bounding the Integrality Gap}
We show that for any $d>0$, and for a sufficiently large $r$, the integrality gap of the instance $G(d,s,s)$ is at least $3/2 - O(1/d)$.
\begin{theorem}
\label{thm:gapexample}
For any $d>0$ and $r\geq d$, the integrality gap of the instance $G(d,s,s)$ is at least $3/2 -8/d$.
\end{theorem}

First, we show that the optimal value of the LP is at most $d/2$.
Define $x^*_a:=1/2$ for all arcs $a\in A(G_d)$. Charikar et al.~\cite{CGK06} show that $x^*$ belongs to the Held-Karp relaxation polytope \cite{HK70}. Since any solution of the Held-Karp relaxation
polytope is a feasible solution to \lpacss\ for $k=1$, $x^*$ is also a feasible solution to \lpacss. Furthermore, since the sum of the cost of the arcs of $G_d$ is $d$, i.e., $c(A(G_d))=d$, we have
$\sum_a c(a)x^*_a = d/2$.  Hence, the optimal value of LP is at most $d/2$.
\begin{lemma}[Charikar et al.~\cite{CGK06}]
For $k=1$, the optimum value of \lpacss~for the graph $G_d$ is at most $d/2$.
\end{lemma}

For any $d>0$, let $H_d$ be the minimum cost strongly connected subgraph of
 $G_d$, and $T(d):=c(A(H_d))$ be the cost of $H_d$. In the rest of the section, we prove the following lemma:
 \begin{lemma}
 \label{lem:optlowerbound}
 For all $d>0$,
 \begin{equation}
 \label{eq:optlowerbound}
 T(d) \geq \frac{3d - 1}{4} - \frac{3d}{r}.
 \end{equation}
 \end{lemma}
Let $H^{(i)}_d:=H_d \cap G^{(i)}_d$ be the $i^{th}$ column of $H_d$.
Observe that $H_d^{(i)}$ can be
incident to (at most) four arcs of the $d^{th}$ level arcs of $H_d$. Let
$$A_d(i):=\left\{(v_i,v_{i-1}),(v_{i+1},v_i),(u_{i-1},u_i),(u_i,u_{i+1})\right\} \cap A(H_d),$$
be the set of those arcs. 
We can lower-bound $c(A(H^{(i)}_d))$ based on the number of arcs that is incident to $H^{(i)}_d$ (note that since $H_d$ is strongly connected, $|A_d(i)| \geq 2$):

\begin{description}

\item[Case 1:] $|A_d(i)|\geq 3$\\
  In this case, we must have
  \begin{equation}
  \label{eq:case1bound}
  c(A(H^{(i)}_d)) \geq T(d-1)/r.
  \end{equation}
  The inequality essentially follows from the fact that $H^{(i)}_d$ is a strongly connected subgraph of $G_{d-1}$. This is because the remaining arcs of the graph, $H_d\setminus H^{(i)}_d$, can only connect (or unify)
  the source and sink of $H^{(i)}_d$, i.e., $u_i$ and $v_i$.
  The $1/r$ factor follows from the fact that the cost of each arc of $G_{d-1}$ is $r$ times the corresponding arc in $G^{(i)}_d$.

\medskip

\item[Case 2:] $|A_d(i)|=2$, and each of $u_i$ and $v_i$ is incident to exactly one arc of $A_d(i)$\\
  Similar to the previous case, here we have
  \begin{equation}
  \label{eq:case2bound}
  c(A(H^{(i)}_d)) \geq T(d-1)/r.
  \end{equation}
	As we will see in Lemma \ref{lem:case2column}, at most two columns of $H_d$ may satisfy this case. Therefore, although we have the worse lower-bound on $c(H^{(i)}_d)$ in this case, it has an insignificant effect on the final lower-bound.

\medskip

\item[Case 3:] $|A_d(i)|=2$, and one of $u_i$ or $v_i$ is incident to none of the arcs of $A_d(i)$\\
  Here we obtain a better lower-bound. For $1\leq j\leq r$, let $H^{(i;j)}_d$ be the $j^{th}$ column of $H^{(i)}_d$ with source $u_{i,j}$ and sink $v_{i,j}$.
  It follows that the only $u_i,v_i$ (or $v_i, u_i$) path in
  $H_d$ is the one that is made by the  $d-1$ level arcs connecting the columns of $H^{(i)}_d$, i.e., $u_i, u_{i,1}, u_{i,2}, \ldots, u_{i,r}, v_i$ (resp. $v_i, v_{i,r}, v_{i,r-1},\ldots,v_{i,1}, u_i$).
 Therefore, $H^{(i)}_d$ must contain all of the $(d-1)^{th}$ level arcs of
  $G^{(i)}_{d}$.
  Since each   column of $H^{(i)}_d$ is incident to 4 arcs of level $(d-1)^{th}$, by repeated application of case 1,
  we obtain
  \begin{eqnarray}
  c(A(H^{(i)}_d))&\geq& 2(r+1)c_{d}(d-1) + \sum_{j=1}^r c(A(H^{(i;j)}_d)) \nonumber \\
  &=& 2(r+1)c_{d}(d-1) + \frac{T(d-2)}{r}.\label{eq:case3bound}
\end{eqnarray}
\end{description}
Next, we show that there are at most $2$ columns satisfying the second case.

\begin{figure}
\begin{center}
\includegraphics[width=4.8in]{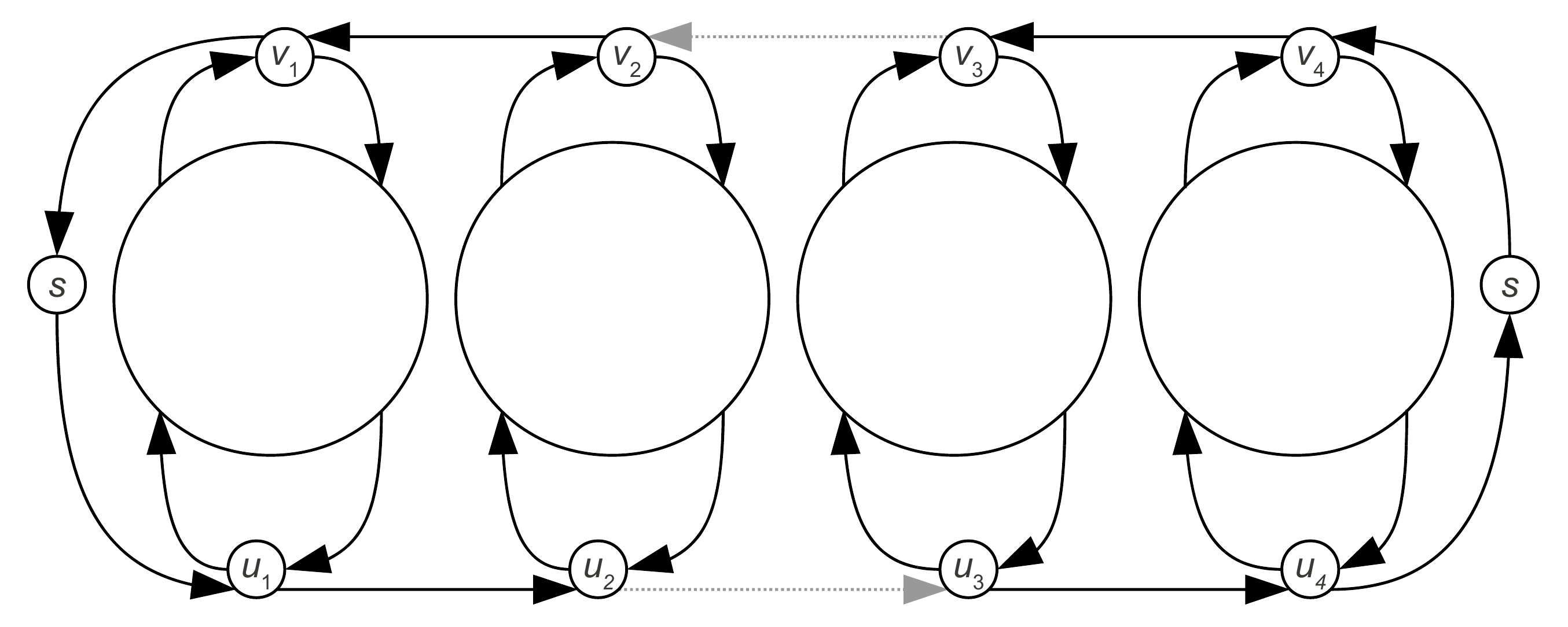}
\end{center}
\label{fig:case2}
\caption{An illustration of $H_d$ where the second column satisfies Case~2. 
The black arcs represent the arcs of $H_d$, and grey arcs represent the removed arcs.
Observe that every arc at level $d$ is a min-cut of $H_d$.
 }
\end{figure}

\begin{lemma}
\label{lem:case2column}
At most two columns of $H_d$ satisfy the second case.
\label{lmm:atmost-two}
\end{lemma}
\begin{proof}
The proof is a simple case analysis argument.
First, observe that there exists a column satisfying the second case in $H_d$ if and only if $(v_{i},v_{i-1}), (u_{i-1},u_{i}) \notin H_d$ for some $1\leq i\leq r+1$.
Now, suppose this is the case.
It then follows that $H_d$ must contain all arcs at level $d$ except these two arcs because each of the other arcs is a min-cut of $H_d$.
See Figure~\ref{fig:case2}.
Therefore, all except (at most) two of the columns of $H_d$ are adjacent to exactly 4 arcs at level $d$.
\end{proof}

Now we are ready to prove Lemma \ref{lem:optlowerbound}.\\

\begin{proofof}{Lemma \ref{lem:optlowerbound}}
We prove by induction. First observe that  $T(0)=0$ and $T(1)=1/2$ satisfying \eqref{eq:optlowerbound}.
Let $N_1, N_2, (r-N_1-N_2)$ be the number of columns satisfying case 1, 2, 3, respectively. We  divide the cost of each arc at level $d$ equally between the columns incident to it.
This incurs a cost of $3 c_d(d)/2$ to the columns satisfying case 1, $c_d(d)$ to the rest of the columns and at least $c_d(d)$ to the source vertex $s$ (note that $s$ is adjacent to
at least two arcs at level $d$).
Using equations \eqref{eq:case1bound}, \eqref{eq:case2bound}, \eqref{eq:case3bound} we get:
\begin{eqnarray*}
T(d) &\geq & c_d(d) + \min_{0\leq N_1,N_2\leq r} \left\{ N_1\left(\frac{3c_d(d)}{2}+\frac{T(d-1)}{r}\right)  +N_2 \left(c_d(d) + \frac{T(d-1)}{r}\right) \right. \\
      &&\left.~~~~~~+(r-N_1-N_2) \left( c_d(d) + 2(r+1) c_{d}(d-1) + \frac{T(d-2)}{r} \right)\right\} \\
      &\geq & \min_{0\leq N\leq r} \left\{N \left(\frac{3 c_d(d)}{2}+ \frac{T(d-1)}{r}\right) \right.\\
 & & \left. ~~~~~~~ +(r-N)\left(c_d(d) + 2(r+1) c_{d}(d-1)+ \frac{T(d-2)}{r}\right)\right\}\\
      &\geq & \min_{0\leq \alpha \leq 1} \left\{\alpha \left(\frac{3 r}{4(r+1)}+ T(d-1)\right) +(1-\alpha)\left(\frac{3r}{2(r+1)} + T(d-2)\right)\right\}\\
      &\geq & \min\left\{ 3/4 + T(d-1), 3/2 + T(d-2)\right\} - 3/r.
\end{eqnarray*}
The second inequality follows from the fact that $N_2\leq 2$.
The third inequality follows from equation \eqref{eq:costfunction}, and the last one follows from a simple algebra.

Now, we may apply the induction hypothesis to $T(d-1)$ and $T(d-2)$. We get
\begin{eqnarray*}
T(d) &\geq &\min\left\{ \frac34 + \frac{3(d-1)-1}{4} - \frac{3(d-1)}{r}, \frac32 + \frac{3(d-2)-1}{4} - \frac{3(d-2)}{r}\right\} - \frac{3}{r} \\
 &\geq & \frac{3d-1}{4} - \frac{3d}{r},
\end{eqnarray*}
which completes the proof.
\end{proofof}

\medskip

This completes the proof of Theorem \ref{thm:gapexample}.

\section{Conclusion}

We  presented a simple $(1+1/k)$-approximation algorithm based on the rounding by sampling method for the minimum size $k$-arc connected subgraph problem.  Unlike recent applications of the rounding by sampling method \cite{AsadpourGMGS10,GharanSS11}, our algorithm has a flavor of the {\em iterated rounding
method} \cite{Jain01} in its particular use of the extreme point solutions.
The main open problem is to find a better than factor  $2$-approximation algorithm for the minimum cost strongly connected subgraph problem.

We also showed that  the integrality gap of the minimum cost strongly connected subgraph problem is at least $1.5-\eps$, for any $\eps>0$.
This leaves an interesting open question whether the lower bound of $1+\Omega(1/k)$ is achievable for the minimum size $k$-arc connected subgraph problem as well.

\bigskip
\noindent{\bf Acknowledgments:}
We thank Joseph Cheriyan for useful discussions on the preliminary construction of the integrality-gap instance.

\bibliographystyle{alpha}
\bibliography{kacss}

\end{document}